\documentclass{ws-ijbc}

\usepackage{graphicx}
\graphicspath{{figures-pdf/}}
\usepackage{multirow,bigstrut}
\usepackage{amssymb,amsmath,bm}

\usepackage[normalem]{ulem}
\usepackage[retainorgcmds]{IEEEtrantools}
\usepackage{xcolor}

\newlength\figwidth
\newlength\imagewidth
\usepackage{color}
\definecolor{dgreen}{rgb}{0,.6,0}

\newtheorem{property}{Property}

\begin{document}

\catchline{}{}{}{}{} 

\markboth{C. Li et al.}{Breaking a chaotic image encryption algorithm based on modulo addition and XOR operation}

\title{Breaking a chaotic image encryption algorithm based on modulo addition and XOR operation}

\author{Chengqing Li\textsuperscript{1}\thanks{Corresponding author, chengqingg@gmail.com},
Yuansheng Liu\textsuperscript{1}, Leo Yu Zhang\textsuperscript{2}, and Michael Z. Q. Chen\textsuperscript{3}\\
\textsuperscript{1} College of Information Engineering,\\ Xiangtan University, Xiangtan 411105, Hunan, China\\[0.5em]
\textsuperscript{2} School of Mathematics and Computational Science,\\
Xiangtan University, Xiangtan 411105, Hunan, China\\
\textsuperscript{3} Department of Mechanical Engineering,\\
The University of Hong Kong, Hong Kong
}

\maketitle


\begin{abstract}
This paper re-evaluates the security of a chaotic image encryption algorithm called MCKBA/ HCKBA and finds that it can be broken efficiently with two known plain-images and the corresponding cipher-images. In addition, it is reported that a previously proposed breaking on MCKBA/HCKBA can be further improved by reducing the number of chosen plain-images from four to two. The two attacks are both based on the properties of solving a composite function involving the carry bit, which is composed of the modulo addition and the bitwise OR operations. Both rigorous theoretical analysis and detailed experimental results are provided.
\end{abstract}


\section{Introduction}

The subtle similarities between chaos and cryptography make chaos considered as a special way to
design secure and efficient encryption schemes \cite{YaobinMao:CSF2004,ChenJY:Joint:TCSII11}. Meanwhile,
some cryptanalysis work demonstrated that some chaos-based encryption schemes are vulnerable to various
conventional attacks from the viewpoint of modern cryptology \cite{LiShujun:YTSCipher:IEEETCASII2004,Xiao:ImproveTable:TCASII06,SolakErcan:AnaFridrich:BAC10,Solak:multichaotic:OC10}.
In addition, some specific security flaws of chaos-based encryption schemes were reported \cite{ZhouJT:Coder:TCSI11,FChen:PeriodArnold:TIT12}.
\cite{AlvarezLi:Rules:IJBC2006} concluded some general approaches to evaluating security of chaos-based encryption schemes.

Due to the simplicity and low computation complexity of bitwise exclusive OR operation and modulo addition, they are widely used in
traditional text encryption schemes and hash functions. Possible generation of carry bit by the modulo addition
makes the two operations are neither identical nor interchangeable. Some properties existing in multi-round combination of the two
basic operations were derived to facilitate differential attacks on some traditional text encryption schemes or searching collision of
hash functions \cite{Paul:Modulo:LNCS05,Wang:Hash:LNCS05}. Among many chaos-based encryption schemes, the two operations are
the basic involved (even only) substitution functions. In \cite{Li:AttackingISNN2005}, \cite{Li:AttackDSEA2006}, \cite{Li:AttackingRCES2008}, and \cite{Li:AttackingIVC2009}, the following properties about $n$-bit integers $\alpha, \beta, \gamma, x, y$ were found to support or enhance the proposed attacks on the corresponding encryption schemes in turn.
\begin{itemlist}
\item If $\alpha\oplus \beta=2^n-1$, then equation $(\alpha\oplus x)=(\beta\oplus x) \dotplus \gamma$ has unique solution modulo $2^{n-1}$,
where $(a\dotplus b)=(a+b)\bmod 2^n$;

\item Equation $|(\alpha\oplus \beta)-(\beta\oplus \gamma)|=|(\alpha\oplus\bar{\beta})-(\beta\oplus\bar{\gamma})|$ always exists;

\item If $\alpha\oplus \beta=\gamma$, then $|\alpha-\beta|\leq \gamma$;

\item If $(((x\dotplus \alpha)\oplus \beta)\dotplus \gamma)\equiv x\oplus y$, then $y\equiv \beta\pmod{2^{n-1}}$.
\end{itemlist}

In 2000, Yen \textit{et al.} proposed a chaotic key-based algorithm (CKBA) by encrypting each pixel of a plain-image by four possible operations: XORing or XNORing
it with one of two predefined sub-keys. The exerted operation is determined by a pseudo-random number sequence (PRNS) generated by iterating the logistic map
\cite{Yen:CKBA:ISCAS2000}. In 2002, S. Li \textit{et al.} broke CKBA with only one known/chosen-image in \cite{Li-Zheng:CKBA:ISCAS2002}. In 2005, Socek \textit{et al.} proposed an enhanced
version of CKBA (ECKBA) employing the following four methods: 1) replacing the logistic map with a piecewise linear chaotic map (PWLCM); 2) increasing the bit length of secret key
to 128; 3) adding a modulo addition and an XOR operation; 4) running all the basic encryption functions multiple times. To achieve a much better balance between encryption load and
security of high level, in 2007 Rao \textit{et al.} proposed a modified version of CKBA (MCKBA) in \cite{Rao:ModifiedCKBA:ICDSP07} by employing a modular addition operation like \cite{SocekLi:SecureComm2005}. To further enhance the security of MCKBA against brute-force attack, in 2010 Gangadhar \textit{et al.} replaced
the logistic map with a simple hyperchaos generator proposed in \cite{Takahashi:HyperChaos:TCASII04} and
names the algorithm HCKBA (Hyper Chaotic-Key Based Algorithm) \cite{CH:HCKBA:IJBC10}. Since the two schemes MCKBA and HCKBA share the same
structure, \cite{LCQ:MCKBA:IJBC11} analyzed them together and reported the following points:
\begin{itemlist}
\item Equivalent secret key of MCKBA/HCKBA can be obtained from four pairs of chosen-plaintexts;

\item Encryption result of MCKBA/HCKBA is not sensitive to changes of plain-image;

\item Encryption result of MCKBA is not sensitive to changes of two sub-keys.

\item The lower bound on the number of queries $(\alpha, \beta)$ to solve unknown variable $x$ in equation
\begin{equation}
y=(\alpha\dotplus x)\oplus (\beta \dotplus x)
\label{eq:essentialfunction}
\end{equation}
in terms of modulo $2^{n-1}$ is 3 if $n\ge 4$.
\end{itemlist}

This paper re-evaluates the security of MCKBA/HCKBA and reports the following points: 1) some properties of Eq.~(\ref{eq:essentialfunction})
are provided to support practical approaches to solving Eq.~(\ref{eq:essentialfunction}); 2) MCKBA/HCKBA can be efficiently broken with two
known-plaintexts; 3) the chosen-plaintext attack proposed in \cite{LCQ:MCKBA:IJBC11} can be further improved and the number of required
chosen-plaintexts is only two.

The rest of this paper is organized as follows. The image encryption algorithm under study is briefly
introduced in Sec.~\ref{sec:scheme}. A known-plaintext attack and an improved chosen-plaintext attack on the algorithm is presented in
Sec.~\ref{sec:cryptanalysis} with experimental results. The last section concludes this paper.

\section{The Chaotic Image Encryption Algorithm Under Study}
\label{sec:scheme}

The encryption object of MCKBA is a gray-scale image of size $M\times
N$ (width$\times$height), which is scanned in the raster order and represented as a one dimensional sequence
$\bm{I}=\{I(i)\}_{i=0}^{MN-1}$. Then, a binary sequence $\bm{I}_b=\{I_b(l)\}_{l=0}^{8MN-1}$ is constructed,
where $\sum_{j=0}^7I_b(8\cdot i+j)\cdot 2^j=I(i)$ for $i=0\sim MN-1$. With a pre-defined integer parameter $n$, an $n$-bit number sequence $\bm{J}=\{J(k)\}_{k=0}^{\lceil 8MN/n\rceil-1}$ is generated, where $J(k)=\sum_{j=0}^{n-1}I_b(n\cdot k+j)\cdot 2^j$. In case $(8MN)$ is not a multiple of $n$,
the sequence $\bm{I}_b$ is padded with some zero bits. Without loss of generality, it is assumed that $n$ can divide $(8MN)$ in this paper. MCKBA operates on the intermediate sequence $\bm{J}$ and obtains
$\bm{J}'=\{J'(k)\}_{k=0}^{8MN/n-1}$, where $J'(k)=\sum_{j=0}^{n-1}I_b'(n\cdot k+j)\cdot 2^j$. Finally, cipher-image
$\bm{I}'=\{I'(k)\}_{k=0}^{MN-1}$ is obtained via $I'(k)=\sum_{j=0}^7I_b'(8\cdot k+j)\cdot 2^j$. Based on the above preliminary introduction, MCKBA is described
with the following four parts\footnote{Since the sole difference between MCKBA and HCKBA is the generator of PRBS, only MCKBA is introduced here
with a concise and consistent form to illustrate the encryption procedure.}.
\begin{itemlist}
\item \textit{The secret key}: Two random numbers $key_1$, $key_2\in\{0, \cdots, 2^n-1\}$, and the initial condition $x(0)\in(0,1)$
of the logistic map
\begin{equation}
x(k+1)=3.9\cdot x(k)\cdot(1-x(k)),\label{equation:Logistic}
\end{equation}
where $\sum_{j=0}^{n-1}(key_{1,j}\oplus key_{2,j})=\lceil n/2\rceil$, $key_1=\sum_{j=0}^{n-1}key_{1,j}\cdot 2^j$, $key_2=\sum_{j=0}^{n-1}key_{2,j}\cdot 2^j$, and
$\oplus$ denotes the eXclusive OR (XOR) operation.

\item \textit{Initialization}: Run Eq.~(\ref{equation:Logistic}) iteratively to generate a sequence $\{x(k)\}_{k=0}^{MN/(2n)-1}$ and
derive a pseudo-random binary sequence (PRBS), $\{b(l)\}_{l=0}^{16MN/n-1}$, from the 32-bit binary representation of elements of
the sequence, namely $x(k)=\sum_{j=1}^{32}b(32\cdot k+j-1)\cdot 2^{-j}$.

\item \textit{Encryption}: For $k=0\sim 8MN/n-1$, encrypt the $k$-th plain-element of $\bm{J}$ via
\begin{equation}
J'(k)=
\begin{cases}
(J(k)\dotplus key_1)\oplus key_1 & \mbox{if }B(k)=3;\\
(J(k)\dotplus key_1)\odot  key_1 & \mbox{if }B(k)=2;\\
(J(k)\dotplus key_2)\oplus key_2 & \mbox{if }B(k)=1;\\
(J(k)\dotplus key_2)\odot  key_2 & \mbox{if }B(k)=0,
\end{cases}
\label{eq:Encrypt}
\end{equation}
where $B(k)=2\cdot b(2k)+b(2k+1)$, and
$a\odot b=\overline{a\oplus b}=a\oplus\bar{b}$.

\item \textit{Decryption}: The decryption procedure is similar to that of the encryption except that Eq.~(\ref{eq:Encrypt}) is replaced by
\begin{equation}
J(k)=
\begin{cases}
(J'(k)\oplus key_1)\dot{-}key_1            & \mbox{if }B(k)=3;\\
(J'(k)\oplus \overline{key_1})\dot{-}key_1 & \mbox{if }B(k)=2;\\
(J'(k)\oplus key_2)\dot{-}key_2            & \mbox{if }B(k)=1;\\
(J'(k)\oplus \overline{key_2})\dot{-}key_2 & \mbox{if }B(k)=0,
\end{cases}
\label{eq:decrypt}
\end{equation}
where $a\dot{-} b=(a-b+2^n)\bmod 2^n$.
\end{itemlist}

\section{Cryptanalysis}
\label{sec:cryptanalysis}

Assume that two plain-images and the corresponding cipher-images encrypted with the same secret key are available, and let  $\bm{J}_1=\{J_1(k)\}_{k=0}^{8MN/n-1}$ and $\bm{J}_2=\{J_2(k)\}_{k=0}^{8MN/n-1}$ denote the corresponding intermediate sequences, respectively. Then, one can assure that the two sequences and the corresponding encrypted results $\bm{J}'_1=\{J_1'(k)\}_{k=0}^{8MN/n-1}$ and $\bm{J}'_2=\{J_2'(k)\}_{k=0}^{8MN/n-1}$
satisfy
\begin{equation}
J_1'(k)\oplus J_2'(k)=
\begin{cases}
(J_1(k)\dotplus key_1)\oplus (J_2(k)\dotplus key_1) & \mbox{if }B(k)\in\{2, 3\};\\
(J_1(k)\dotplus key_2)\oplus (J_2(k)\dotplus key_2) & \mbox{if }B(k)\in\{0, 1\}.
\end{cases}
\label{eq:DifferentialEquation}
\end{equation}
No matter what the value of $B(k)$ is, the above equation can be represented in the form of Eq.~(\ref{eq:essentialfunction}). In this section, we first present some properties of the kernel function (\ref{eq:essentialfunction}) on obtaining its solution and then illustrate how to obtain an equivalent secret key of MCKBA/HCKBA with two known plain-images and two chosen plain-images, respectively.

\subsection{Some properties of the kernel function}

\begin{property}
Equivalent form of Eq.~(\ref{eq:essentialfunction})
\begin{equation}
\tilde{y}=y\oplus \alpha\oplus \beta=(\alpha\dotplus x)\oplus(\beta\dotplus x)\oplus \alpha\oplus \beta
\label{eq:essentialfunctionform}
\end{equation}
can be represented as an iteration form
\begin{equation}
\left\{
\begin{IEEEeqnarraybox}[][c]{rCl}
\tilde{y}_{i+1} & = & c_{i+1}\oplus \tilde{c}_{i+1},\\
c_{i+1}         & = & (x_i\cdot \alpha_i) \oplus (x_i\cdot c_i) \oplus (\alpha_i\cdot c_i),\\
\tilde{c}_{i+1} & = & (x_i\cdot \beta_i) \oplus (x_i\cdot \tilde{c}_i) \oplus (\beta_i\cdot \tilde{c}_i),
\end{IEEEeqnarraybox}\right.
\label{eq:bitdecomposition}
\end{equation}
where $c_0\equiv0$, $\tilde{c}_{0}\equiv0$, $x=\sum_{i=0}^{n-1}x_i\cdot 2^i$, $\alpha=\sum_{i=0}^{n-1}\alpha_i\cdot 2^i$, $\beta=\sum_{i=0}^{n-1}\beta_i\cdot 2^i$, $\tilde{y}=\sum_{i=0}^{n-1}\tilde{y}_i\cdot 2^i$ (These notations are the same hereinafter.).
\end{property}
\begin{proof}
Set $c_0=0$, one can calculate
$c_{i+1}$ from $c_i$ and $\alpha_i$ via
\begin{equation}
c_{i+1}=(x_i\cdot \alpha_i) \oplus (x_i\cdot c_i) \oplus (\alpha_i\cdot c_i)
\label{eq:getcarrybit}
\end{equation}
for $i=0\sim n-2$. So, $c_{i+1}$ denote the carry bit generated by $x$ and $\alpha$ in the $i$-th bit plane. Set $\tilde{c}_0=0$, one can then obtain
\begin{equation*}
\tilde{c}_{i+1}=(x_i\cdot \beta_i) \oplus (x_i\cdot \tilde{c}_i) \oplus (\beta_i\cdot \tilde{c}_i)
\end{equation*}
for $i=0\sim n-2$. Similarly, $\tilde{c}_{i+1}$ denote the carry bit generated by $x$ and $\beta$ in the $i$-th bit plane. Obviously, $\tilde{y}_0=(\alpha_0\oplus x_0)\oplus(\beta_0\oplus x_0)\oplus \alpha_0\oplus \beta_0\equiv 0$.
Then, the $(i+1)$-th bit plane of Eq.~(\ref{eq:essentialfunctionform}) can be represented as
\begin{eqnarray*}
\tilde{y}_{i+1} & = & (\alpha_{i+1}\oplus c_{i+1}\oplus x_{i+1})\oplus (\beta_{i+1}\oplus\tilde{c}_{i+1}\oplus x_{i+1})\oplus \alpha_{i+1}\oplus \beta_{i+1}\\
                & = & c_{i+1}\oplus \tilde{c}_{i+1},
\end{eqnarray*}
where $i=0\sim n-2$. So, $\tilde{y}_i$ can be easily calculated iteratively according to Eq.~(\ref{eq:bitdecomposition}) for $i=0\sim n-2$, which
can also be done via checking Table~1 listing the values of $\tilde{y}_{i+1}$ under all possible different values of $\alpha_i, \beta_i,\tilde{y}_{i}, x_i$, and $c_i$.
\begin{table}[htbp]
\tbl{The values of $\tilde{y}_{i+1}$ corresponding to the values of $\alpha_i, \beta_i,\tilde{y}_{i}, x_i$, and $c_i$.}
{\begin{tabular}{*{8}{c}c}
\toprule
\multirow{2}{0.3in}{$(x_i, c_i)$}     & \multicolumn{8}{c}{$(\alpha_i, \beta_i,\tilde{y}_{i} )$} \\
\cline{2-9} & $(0,0,0)$ & $(0,0,1)$ & $(0,1,0)$ & $(0,1,1)$ & $(1,0,0)$ & $(1,0,1)$ & $(1,1,0)$ & $(1,1,1)$\\\hline
(0, 0)      &     0     &   0       &   0       &   1       &   0       &   0       &       0   &   1    \\
(0, 1)      &     0     &   0       &   1       &   0       &   1       &   1       &       0   &   1    \\  \hline
(1, 0)      &     0     &   1       &   1       &   1       &   1       &   0       &       0   &   0    \\
(1, 1)      &     0     &   1       &   0       &   0       &   0       &   1       &       0   &   0    \\ [-4pt]
\botrule
\end{tabular}}
\label{table:carrybits}
\end{table}
\end{proof}

\begin{property}
Given $(\alpha_i, \beta_i, \tilde{y}_{i}, \tilde{y}_{i+1})$, no information about $x_i$, $c_i$ and $\tilde{c}_i$ can be obtained (Note that $c_0$ and $\tilde{c}_0$ are excluded since they are pre-defined constants.) if and only if $(4\alpha_i+2\beta_i+\tilde{y}_i)\in \{0, 6\}$.
\end{property}
\begin{proof}
Since only the data in the 0, 6-th column (zero-based) of Table~1 are identical, it is impossible to obtain any information about $x_i$, $c_i$ and $\tilde{c}_i$ from $(\alpha_i, \beta_i, \tilde{y}_{i}, \tilde{y}_{i+1})$ if and only if $(4\alpha_i+2\beta_i+\tilde{y}_i)\in \{0, 6\}$.
\end{proof}

\begin{property}
Given $(\alpha_i, \beta_i, \tilde{y}_{i}, \tilde{y}_{i+1})$, the unknown bit $x_i$ can be determined via $x_i=\alpha_i\oplus\tilde{y}_{i+1}$, if and only if
$(4\alpha_i+2\beta_i+\tilde{y}_i)\in \{1, 7\}$.
\end{property}
\begin{proof}
Only under the cases shown in the 1, 7-th column of Table~1, $x_i$ can be determined by $(\alpha_i, \beta_i, \tilde{y}_{i}, \tilde{y}_{i+1})$ without knowledge of $c_i$. It is easy to verify that $x_i=\alpha_i\oplus\tilde{y}_{i+1}$ in terms of value.
\end{proof}

\begin{property}
Given $(\alpha_i, \beta_i, \tilde{y}_i, \tilde{y}_{i+1})$, carry bits $c_i$ and $\tilde{c}_i$ can be determined via $c_i=\beta_i\oplus \tilde{y}_{i+1}$ and
$\tilde{c}_i=\beta_i\oplus \tilde{y}_{i+1}\oplus \tilde{y}_{i}$ if and only if
$(4\alpha_i+2\beta_i+\tilde{y}_i)\in \{3, 5\}$.
\end{property}
\begin{proof}
Only under the cases shown in the 3, 5-th column of Table~1, $c_i$ can be determined by $(\alpha_i, \beta_i, \tilde{y}_{i}, \tilde{y}_{i+1})$ without knowledge of $x_i$. It is easy to verify that $c_i=\beta_i\oplus \tilde{y}_{i+1}$ in terms of value. Then, one can obtain $\tilde{c}_i=\beta_i\oplus \tilde{y}_{i+1}\oplus \tilde{y}_{i}$ since $\tilde{c}_i=c_i\oplus \tilde{y}_{i}$.
\end{proof}

\begin{property}
Given $(\alpha_i, \beta_i, \tilde{y}_{i}, \tilde{y}_{i+1})$, the scope of the unknown bits $x_i, c_i$ can be narrowed via
\begin{equation}
(x_i, c_i)\in
\begin{cases}
\{(0, 0), (1, 1)\} & \mbox{if } \tilde{y}_{i+1}=0;\\
\{(0, 1), (1, 0)\} & \mbox{if } \tilde{y}_{i+1}=1,
\end{cases}
\label{eq:cases2_4}
\end{equation}
if and only if $(4\alpha_i+2\beta_i+\tilde{y}_i)\in \{2, 4\}$.
\end{property}
\begin{proof}
Referring to the cases shown in the 2, 4-th column of Table~1, the scope of $(x_i, c_i)$ can be narrowed according to value of $\tilde{y}_{i+1}$.
It is easy to obtain Eq.~(\ref{eq:cases2_4}) from Table~1. Therefore, the ``if" part of the property is proven. Note that the number of possible values of $(4\alpha_i+2\beta_i+\tilde{y}_i)$ is only eight, and the sufficient and necessary conditions on obtaining different information on $x_i$ and $c_i$ under other six cases have been presented. Therefore, the ``only if" part of the property is also proven.
\end{proof}

\begin{property}
If $(\alpha_{i-1}\oplus \beta_{i-1})=1$ and $\tilde{y}_i=1$, one has
$c_i=\alpha_{i-1}$, where $i\in\{1, 2, \cdots,  n-2\}$.
\end{property}
\begin{proof}
When $(\alpha_{i-1}, \beta_{i-1})=(1, 0)$ and $\tilde{y}_i=1$, one can get $c_i=1$ from Eq.~(\ref{eq:bitdecomposition}) no matter the value
of $x_{i-1}$. Similarly, one has $c_i=0$ when $(\alpha_{i-1}, \beta_{i-1})=(0, 1)$ and $\tilde{y}_i=1$. So, the property is proven.
\end{proof}

\begin{property}
Given $\alpha, \beta, \tilde{y}$, some bits among the $(n-1)$ least significant bits of $x$ in Eq.~(\ref{eq:essentialfunctionform}) can be determined from the least significant bit to the most significant one.
\label{propsition:solution}
\end{property}
\begin{proof}
The concrete approaches to solving Eq.~(\ref{eq:essentialfunction}) and determining the carry bits can be divided into the following two classes of operations.
\begin{itemlist}
\item \textit{Obtaining information on $x_0$ and $c_1$}:
According to how much information on $x_0$ and $c_1$ can be obtained, $(\alpha_0, \beta_0, \tilde{y}_0)$ is further classified as the following two cases.
\begin{itemlist}
\item $(4\alpha_0+2\beta_0+\tilde{y}_0)\in\{0, 6\}$:
Referring to Property~2, $x_0$ can not be determined in this case, but one can obtain $c_1=0$ if $\alpha_0=0$.

\item $(4\alpha_0+2\beta_0+\tilde{y}_0)\in\{2, 4\}$: As $c_0=0$, one can obtain
\begin{equation}
x_0=
\begin{cases}
0 & \mbox{if } \tilde{y}_{1}=0;\\
1 & \mbox{if } \tilde{y}_{1}=1,
\end{cases}
\label{eq:x0}
\end{equation}
from Eq.~(\ref{eq:cases2_4}). Then, one can further obtain
\begin{equation*}
c_1=
\begin{cases}
0 & \mbox{if } \tilde{y}_{1}=0;\\
1 & \mbox{if } \alpha_0=1 \mbox{ and }\tilde{y}_{1}=1;\\
0 & \mbox{if } \alpha_0=0 \mbox{ and }\tilde{y}_{1}=1.
\end{cases}
\end{equation*}
\end{itemlist}

\item \textit{Obtaining information on $x_i$, $c_i$ and $c_{i+1}$ for $i=1\sim n-2$}:
According to how much information on $x_i$, $c_i$ and $c_{i+1}$ can be obtained by checking $(\alpha_i, \beta_i, \tilde{y}_i)$ and the obtained information on $c_i$ for $i=1\sim n-2$ in order, $(\alpha_i, \beta_i, \tilde{y}_i)$ is categorized as the following four cases\footnote{As confirmation of $c_i$ is equivalent to that of $\tilde{c}_i$, the latter is not mentioned.}.
\begin{itemlist}
\item $(4\alpha_i+2\beta_i+\tilde{y}_i)\in\{0, 6\}$:
Referring to Property~2, no information on $x$ can be determined in this case.
The value of $c_{i+1}$ can be determined by Eq.~(\ref{eq:getcarrybit}) if
\begin{equation}
\{c_i+\alpha_i\}\in \{0, 2\}\mbox{ is known.}
\label{eq:condition}
\end{equation}

\item $(4\alpha_i+2\beta_i+\tilde{y}_i)\in\{1, 7\}$: One has
\begin{equation}
x_i=\alpha_i\oplus \tilde{y}_{i+1}
\label{eq:xi17}
\end{equation}
from Property~3. If $c_i$ has been determined, one can obtain
$c_{i+1}$. Even $c_i$ is still unknown, one can confirm $c_{i+1}$ by Eq.~(\ref{eq:getcarrybit}) if $(\alpha_i+x_i)=0$ or $(\alpha_i+x_i)=2$ is known.

\item $(4\alpha_i+2\beta_i+\tilde{y}_i)\in\{2, 4\}$: If $c_i$ has been determined, based on Property~5 one can obtain
\begin{equation}
x_i=
\begin{cases}
1-\tilde{y}_{i+1} & \mbox{if }  c_i=1;\\
\tilde{y}_{i+1}  & \mbox{if }  c_i=0,
\end{cases}
\label{eq:xi24}
\end{equation}
and further confirm the value of $c_{i+1}$.

\item $(4\alpha_i+2\beta_i+\tilde{y}_i)\in\{3, 5\}$: Referring to Property~4, one can obtain $c_i=\beta_i\oplus \tilde{y}_{i+1}$.
\end{itemlist}
\end{itemlist}
\end{proof}

\subsection{Known-plaintext attack}

Known-plaintext attack is one of the classic attack models where the attacker (or cryptanalyst) can
access both some plaintexts and the corresponding encryption results encrypted with the same secret key.
In \cite[Sec.~3.2]{CH:HCKBA:IJBC10}, the original authors claimed that HCKBA has strong vulnerability against known-plaintext attack.
However, we found MCKBA/HCKBA is very weak against the attack, which is supported by the properties of Eq.~(\ref{eq:essentialfunction}) shown
in the previous subsection.

Under the scenario of known-plaintext attack, breaking MCKBA/HCKBA is to determine some information of its equivalent secret key, $key1$, $key2$ and $\{B(k)\}_{k=0}^{8MN/n-1}$,
by solving Eq.~(\ref{eq:DifferentialEquation}) and utilizing some properties of MCKBA/HCKBA. From Property~\ref{propsition:solution}, one can see that some bits of $key1$ and $key2$ can be obtained from Eq.~(\ref{eq:DifferentialEquation}) for any $k\in\{0, \cdots, 8MN/n-1\}$, where the other unknown bits are just set as zero. Let $key(k)$ denote the obtained solution of Eq.~(\ref{eq:DifferentialEquation}) and $s(k,i)$ represent $key(k)_i$ is confirmed definitely or not, i.e. set $s(k,i)=1$ if $key(k)_i$ is confirmed by Eq.~(\ref{eq:x0}), Eq.~(\ref{eq:xi17}), or Eq.~(\ref{eq:xi24}); otherwise set $s(k,i)=0$, where
$key(k)=\sum_{i=0}^{n-1}key(k)_i\cdot 2^i$, and $k=0\sim 8MN/n-1$. Then, one may reconstruct set $\{key1, key2\}$ from $\{key(k)\}_{k=0}^{8MN/n-1}$ and
$\{s(k,i)\}_{k=0, i=0}^{8MN/n-1, n-2}$ by identifying and combining the known bits belonging to the same number, which is described by the following steps.
\begin{itemlist}
\item \textit{Step 1)}: Set $\mathbb{K}=\{key(0), key(1), \cdots, key(8MN/n-1)\}$.
Delete some elements of $\mathbb{K}$ to assure that every pair of elements of $\mathbb{K}$ has at least one different confirmed bit.

\item \textit{Step 2)}: Search for the first two elements in $\mathbb{K}$ whose
number of confirmed bits are most but the confirmed bits of the two elements are not all the same. Let $Seed(0)$ and $Seed(1)$
denote the two seed elements and delete them from $\mathbb{K}$.

\item \textit{Step 3)}: Check each element of $\mathbb{K}$ in turn and do the following two operations if it has one confirmed bit which is different from that of $Seed(i)$:
1) update $Seed(1-i)$ by combining all the confirmed bits of the element into that of $Seed(1-i)$; 2) delete the element from $\mathbb{K}$, where
$i\in\{0, 1\}$.

\item \textit{Step 4)}: Repeat \textit{Step 3)} iteratively till the numbers of confirmed bits of $Seed(0)$ and $Seed(1)$ are not increased in the whole step.

\item \textit{Step 5)}: Terminate the whole search operation when all bits of $Seed(0)$ and $Seed(1)$ are confirmed bits; otherwise
repeat \textit{Step 2)} through \textit{Step 4)} till the cardinality of $\mathbb{K}$ is less than 2.
\end{itemlist}

Let us study the probability on obtaining $x_i$ and $c_i$ with one pair of $\alpha$, $\beta$ and $\tilde{y}$ under assumption that $\alpha$, $\beta$ and $x$ distributes over $\{0, \cdots, 2^n-1\}$ uniformly. First, one has $Prob(c_0=1)=0$ and
\begin{equation*}
Prob(c_i=1)=\frac{3}{4}Prob(c_{i-1}=1)+\frac{1}{4}Prob(c_{i-1}=0)
\end{equation*}
for $i=1\sim n-1$. Solve the above iteration function, one can obtain $Prob(c_i=1)=\frac{2^i-1}{2^{i+1}}$. Obviously, one has $Prob(\tilde{y}_0=0)=1$.
Observe Table~1, one can see that the value of $\tilde{y}_i$ is determined by the values of $\tilde{y}_{i-1}$, $c_{i-1}$, $x_{i-1}$, $\alpha_{i-1}$ and $\beta_{i-1}$ for $i=1\sim n-1$. Considering all possible cases, one has
\begin{IEEEeqnarray*}{rcl}
Prob(\tilde{y}_i=0)
& = & Prob(\tilde{y}_{i-1}=0) \left(Prob(c_{i-1}=0)\left(\frac{1}{2}\cdot 1+\frac{1}{2}\cdot \frac{1}{2}\right)+Prob(c_{i-1}=1)\left(\frac{1}{2}\cdot \frac{1}{2}+\frac{1}{2}\cdot 1\right)\right)\\
&  &+Prob(\tilde{y}_{i-1}=1)\left(Prob(c_{i-1}=0)\left(\frac{1}{2}\cdot \frac{1}{2}+\frac{1}{2}\cdot \frac{1}{2}\right)+Prob(c_{i-1}=1)\left(\frac{1}{2}\cdot \frac{1}{2}+\frac{1}{2}\cdot \frac{1}{2}\right)\right)\\
&  &=\frac{3}{4}Prob(\tilde{y}_{i-1}=0)+\frac{1}{2}Prob(\tilde{y}_{i-1}=1)
\end{IEEEeqnarray*}
for $i=1\sim n-1$. Solve the iteration function, one can obtain
\begin{equation}
Prob(\tilde{y}_i=0)=\frac{2}{3}+\frac{1}{3\cdot 4^i}.
\label{eq:probyi}
\end{equation}
From the proof of Property~7, one can first calculate $Prob[c_0]=1$, $Prob[c_1]=Prob((4\alpha_0+2\beta_0+\tilde{y}_0)\in\{0, 6\})\cdot\frac{1}{2}+Prob((4\alpha_0+2\beta_0+\tilde{y}_0)\in\{2, 4\})\cdot 1=\frac{1}{4}+\frac{1}{2}=\frac{3}{4}$
and
\begin{IEEEeqnarray}{rcl}
Prob[c_i]  & = & Prob((4\alpha_i+2\beta_i)\in\{0, 6\})Prob(\tilde{y}_{i-1}=0) Prob[c_{i-1}]Prob(\mbox{Condition}~(\ref{eq:condition}) \mbox{ holds})\nonumber\\
&   & +Prob((4\alpha_i+2\beta_i)\in\{0, 6\})Prob(\tilde{y}_{i-1}=1)\left(Prob[c_{i-1}]+(1-Prob[c_{i-1}])Prob((\alpha_i+x_i)\in\{0, 2\})\right)\nonumber\\
          &   &+Prob((4\alpha_i+2\beta_i)\in\{2, 4\})Prob(\tilde{y}_{i-1}=0) Prob[c_{i-1}]+Prob((4\alpha_i+2\beta_i)\in\{2, 4\})Prob(\tilde{y}_{i-1}=1)\nonumber\\
& = &\frac{1}{2} Prob(\tilde{y}_{i-1}=0) Prob[c_{i-1}]\frac{1}{2}+\frac{1}{2}Prob(\tilde{y}_{i-1}=1)\cdot\left(Prob[c_{i-1}]+(1-Prob[c_{i-1}])\cdot\frac{1}{2}\right)\nonumber\\
          &   &+\frac{1}{2}Prob(\tilde{y}_{i-1}=0) Prob[c_{i-1}]+\frac{1}{2}Prob(\tilde{y}_{i-1}=1)\nonumber\\
          & = & Prob[c_{i-1}]\left(\frac{1}{2}Prob(\tilde{y}_{i-1}=0)+\frac{1}{4}\right)+\frac{3}{4}Prob(\tilde{y}_{i-1}=1)\nonumber\\
          & = & Prob[c_{i-1}]\left(\frac{7}{12}+\frac{1}{6\cdot 4^{i-1}}\right)+\frac{1}{4}-\frac{1}{4^i}  \label{eq:probci}
\end{IEEEeqnarray}
for $i=2\sim n-1$, where $Prob[a]$ denotes the probability that the bit $a$ can be confirmed.
Finally, one has $Prob[x_0]=\frac{1}{2}$ and
\begin{IEEEeqnarray}{rcl}
Prob[x_i] &=& \frac{1}{2}Prob(\tilde{y}_{i}=1)+\frac{1}{2}Prob(\tilde{y}_{i}=0) Prob[c_i]
     \label{eq:probxi}
\end{IEEEeqnarray}
for $i=1\sim n-2$. Incorporate Eq.~(\ref{eq:probyi}) and Eq.~(\ref{eq:probci}) into Eq.~(\ref{eq:probxi}), one can obtain that
$Prob[x_0]=\frac{1}{2}$, $Prob[x_1]=0.4062$, $Prob[x_2]= 0.3818$, and $Prob[x_i]\approx 0.37$ for $i\ge 3$. Now, one can assure that $key1_i$ and $key2_i$ can not be confirmed definitely with a probability smaller than or equal to $(1-0.37)^{n_0}=0.63^{n_0}$ and $0.63^{(8MN/n-n_0)}$ respectively, where $n_0$ is cardinality of the set $\{k | B(k)\in\{2, 3\}, k=0\sim 8MN/n-1\}$. Note that
confirmation of the bits of $x$ in  Eq.~(\ref{eq:essentialfunction}) is determined by $\alpha$, $\beta$, and the used unknown variable $x$, and no any bit can be confirmed under some cases (See the examples shown in Table~2). In addition, as for plaintext chose from natural images, $\alpha$ follows Gaussian distribution. But, we can still believe that
the set $\{key1, key2\}$ can be reconstructed in a high probability since the values of $n_0$ and $8MN/n-n_0$ are both very large in general.
\begin{table}[htbp]
\tbl{The number of confirmed bits of $x$ in Eq.~(\ref{eq:essentialfunction}) under some sets of $(\alpha, \beta, \tilde{y})$ when $n=8$, where
$\tilde{y}$ is determined by $(\alpha, \beta, x)$ via Eq.~(\ref{eq:essentialfunction}).}
{\begin{tabular}{c|*{7}{c}c}
\toprule
\multirow{2}{0.1in}{$x$} & \multicolumn{8}{c}{$(\alpha, \beta)$} \\
\cline{2-9} & $(0, 1)$ & $(17, 54)$ & $(31, 102)$ & $(44,93)$ & $(51, 95)$ & $(73, 79)$ & $(87, 122)$ & $(125, 126)$\\\hline
    7       &  4       &   1        &   3         & 4      &    1          &   0    &       6           &  1   \\
    8       &  1       &   3        &   1         &  1    &     1         &    0   &        3          &   2  \\
    9       &  2       &   1        &   3         &  2    &    1          &    1   &        3          &   1  \\
    28      &  1       &   7        &   1         &   1   &    2          &    1   &        2          &   2  \\
    59      &  3       &   4        &   6         &  7    &    3          &    1   &        3          &   1  \\
    71      &  8       &   1        &   3         & 4     &    1          &    0   &        6          &   1  \\
    76      &  1       &   6        &   1         &  1    &    1          &    1   &        2          &   2  \\
    95      &  6       &   2        &   5         & 5     &    0          &   0    &        6          &   1  \\
    99      &  3       &   3        &   3         & 4     &    1          &   2    &        4          &   1  \\
    111     &  5       &   1        &   4         & 1     &    0          &   0    &        5          &   1  \\
    125     &  2       &   4        &   7         & 2     &    0          &   1    &        2          &   1  \\
    127     &  7       &   1        &   6         & 6     &    0          &   0    &        5          &   1  \\[-4pt]
    \botrule
\end{tabular}}
\label{table:numbersolution}
\end{table}

Referring to Proposition~\ref{propsition:lsb} and Eq.~(\ref{eq:Encrypt}), one can obtain the scope of $B(k)$,
\begin{equation}
\mathbb{B}(k)=
\begin{cases}
\{1, 3\}  & \mbox{if } (J_1'(k)\oplus J_1(k))\bmod 2=0;\\
\{0, 2\}  & \mbox{otherwise},
\end{cases}
\label{eq:determinebits2}
\end{equation}
for $k=0\sim 8MN/n-1$.

\begin{proposition}
Assume that $a$ and $x$ are both $n$-bit integers and $n\in \mathbb{Z}^+$,
$((a\dotplus x)\oplus x)$ has the same parity as $a$ and $((a\dotplus x)\odot x)$ has opposite parity as $a$.
\label{propsition:lsb}
\end{proposition}
\begin{proof}
Existence of four equations
\begin{eqnarray*}
((1+x_0)\bmod 2)\oplus x_0 & \equiv & 1,\\
((0+x_0)\bmod 2)\oplus x_0 & \equiv & 0,\\
((1+x_0)\bmod 2)\odot x_0 &  \equiv & 0,\\
((0+x_0)\bmod 2)\odot x_0 &  \equiv & 1,
\end{eqnarray*}
is independent of $x_0$, so the proposition is proved.
\end{proof}

\begin{proposition}
Assume that $a$ and $x$ are both $n$-bit integers, $n\in \mathbb{Z}^+$, one has the following two equations
\begin{eqnarray*}
(a\oplus x)\dot{-}x            & = & (a\oplus x\oplus 2^{n-1})\dot{-}(x\oplus 2^{n-1}),           \label{eq:EquiKey1}\\
(a\oplus \overline{x})\dot{-}x & = & (a\oplus \overline{x\oplus 2^{n-1}})\dot{-}(x\oplus 2^{n-1}).\label{eq:EquiKey2}
\end{eqnarray*}
\label{propsition:msb}
\end{proposition}
\begin{proof}
See the proof of Proposition~1 in \cite{LCQ:MCKBA:IJBC11}.
\end{proof}

According to the pre-defined condition $key1\neq key2$, there are only two possible combinations of $key1$ and $key2$.  Let $(key1^*, key2^*)$ denote the searched version of $(key1, key2)$. Proposition~\ref{propsition:msb} illustrates that the unknown most significant bit of $key1^*=\sum_{j=0}^{n-1}key1^*_j\cdot 2^j$ and $key2^*=\sum_{j=0}^{n-1}key1_j^*\cdot 2^j$ has no influence on decryption of MCKBA/HCKBA. Then, one can further obtain the approximate value of $B(k)$, $B^*(k)$, and the $n-1$ least significant bits of
$key1$ and $key2$ by the following two different ways:
\begin{itemlist}
\item  \textit{W1)} For $k=0\sim 8MN/n-1$, one has
\begin{equation}
B^*(k)=
\begin{cases}
3  & \mbox{if } F(key_1^*, J_1(k), J'_1(k))=0, F(key_2^*, J_1(k), J'_1(k))\neq 0\\
   &  \mbox{or } F(key_1^*, J_2(k), J'_2(k))=0, F(key_2^*, J_2(k), J'_2(k))\neq 0;\\
2  & \mbox{if } G(key_1^*, J_1(k), J'_1(k))=0, G(key_2^*, J_1(k), J'_1(k))\neq0\\
   &  \mbox{or } G(key_1^*, J_2(k), J'_2(k))=0, G(key_2^*, J_2(k), J'_2(k))\neq0;\\
1  & \mbox{if } F(key_1^*, J_1(k), J'_1(k))\neq0, F(key_2^*, J_1(k), J'_1(k))=0\\
   &  \mbox{or } F(key_1^*, J_2(k), J'_2(k))\neq0, F(key_2^*, J_2(k), J'_2(k))= 0;\\
0  & \mbox{if } G(key_1^*, J_1(k), J'_1(k))\neq, G(key_2^*, J_1(k)^*, J'_1(k))=0\\
   &  \mbox{or } G(key_1^*, J_2(k), J'_2(k))\neq0, G(key_2^*, J_2(k), J'_2(k))=0,
\end{cases}
\label{eq:determineBk}
\end{equation}
where $F(x, \alpha, y)=y-((\alpha\dotplus x)\oplus x)$, $G(x, \alpha, y)=y-((\alpha\dotplus x)\odot x)$. Note that
Eq.~(\ref{eq:determinebits2}) makes only two conditions in Eq.~(\ref{eq:determineBk}) need being
verified. Obviously, one can assure that $\sum_{j=0}^{n-2}key^*_{1,j}\cdot 2^j=\sum_{j=0}^{n-2}key_{1,j}\cdot 2^j$, $\sum_{j=0}^{n-2}key^*_{2,j}\cdot 2^j=\sum_{j=0}^{n-2}key_{2,j}\cdot 2^j$.


\item \textit{W2)} When there exists $i\in\{0, \cdots, n-2\}$ satisfying that $s(k, i)=1$, one can obtain
\begin{equation}
\mathbb{B}^*(k)=
\begin{cases}
\{2,3\}  & \mbox{if } key(k)_i\neq key2^*_{i};\\
\{0,1\}  & \mbox{if } key(k)_i\neq key1^*_{i},
\end{cases}
\label{eq:determinebits1}
\end{equation}
for $k=0\sim 8MN/n-1$. Then, the value of $B(k)$ can be obtained by setting $B^*(k)=\mathbb{B}^*(k)\bigcap \mathbb{B}(k)$ for
$k=0\sim 8MN/n-1$. Finally, one can conclude that $(key1^*, key2^*)=(\sum_{i=0}^{n-2}key1^*_{i}\cdot 2^i$, $\sum_{i=0}^{n-2}key2^*_{i}\cdot 2^i)$, and
$\{B^*(k)\}_{k=0}^{8MN/n-1}$ can work together as equivalent secret key of MCKBA/HCKBA due to that $(key1, key2, B(k))=(a, b, c)$
and $(key1, key2, B(k))=(b, a, (c+2)\bmod 4)$ are equivalent for Eq.~(\ref{eq:decrypt}).
\end{itemlist}

Now, let's study the success probability of the above two methods. The success of the method \textit{W1) depends on existence of one of the eight
condition in Eq.~(\ref{eq:determineBk}). As $F(x, \alpha, y)=0$ if and only if $G(x, \alpha, y\oplus (2^n-1))=0$, only one of the two functions
need being studed. Obviously, }
\begin{equation}
\begin{cases}
F(x, 0, y)\equiv y, \\
F(x, 2^{n-1}, y)\equiv y-2^{n-1}.
\end{cases}
\label{eq:SpecialCondition}
\end{equation}
So, $B(k)$ can not be confirmed by  Eq.~(\ref{eq:determineBk}) when
 $J_1(k)\in\{0, 2^{n-1} \}$ and $J_2(k)\in\{0, 2^{n-1} \}$ exist at the same time. It is very hard to derive the probability for other cases theoretically.
Instead, we calculate the probability that $F(x_1, \alpha, y)=F(x_2, \alpha, y)=0$ via simulation, where
\begin{equation}
\sum_{j=0}^{n-1}(x_{1,j}\oplus x_{2,j})=\lceil n/2\rceil,
\label{eq:constraint}
\end{equation}
$x_1=\sum_{j=0}^{n-1}x_{1,j}\cdot 2^j$, $x_2=\sum_{j=0}^{n-1}x_{2,j}\cdot 2^j$. Assume $key1, key2$ distributes uniformly, the distribution of
probability $F(key1, \alpha, y)=F(key2, \alpha, y)=0$ under different values of $\alpha$ and some values of $n$ is shown in Fig.~\ref{fig:probability}.

\begin{figure}[!htb]
\centering
\begin{minipage}{2\figwidth}
\centering
\includegraphics[width=\textwidth]{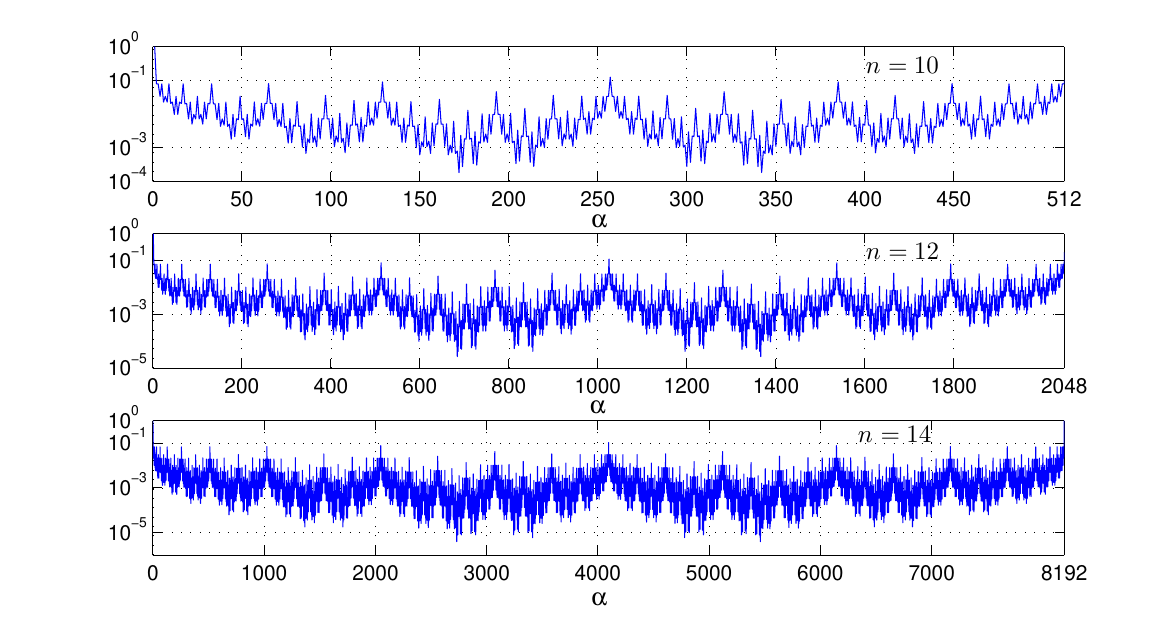}
\end{minipage}
\caption{The probability $F(x_1, \alpha, y)=F(x_2, \alpha, y)=0$ under different value of $\alpha$, where $x_1, x_2$ satisfy the constraint condition (\ref{eq:constraint}).}
\label{fig:probability}
\end{figure}

The success of the method \textit{W2) can be analyzed as follows. Assume $key1, key2$ distributes uniformly, one can get the probability that at least one bit (exclude the most significant bit) of $key(k)$
satisfy $s(k, i)=1$ and one condition of Eq.~(\ref{eq:determinebits1}) is}
\begin{IEEEeqnarray*}{rcl}
Prob[ \mbox{Eq.}~(\ref{eq:determinebits1}) \mbox{ holds} ] & =& 1-\sum_{i=0}^{n-1}\binom{n-1}{i}\left(Prob[x_i]\right)^i\left(1-Prob[x_i]\right)^{n-1-i}\left(\frac{1}{2}\right)^i\\
& \ge & 1-\sum_{i=0}^{n-1}\binom{n-1}{i}\left(\frac{1}{2}\right)^i\left(1-\frac{1}{2}\right)^{n-1-i}\left(\frac{1}{2}\right)^i\\
     & = & 1-\left(\frac{1}{2}\right)^{n-1}\sum_{i=0}^{n-1}\binom{n-1}{i}\left(\frac{1}{2}\right)^i\\
     & = & 1-\left(\frac{3}{4}\right)^{n-1}.
\end{IEEEeqnarray*}
From the above analysis, one can see that $B(k)$, $k=0\sim 8MN/n-1$, can be determined with a probability larger than $1-\left(\frac{3}{4}\right)^{n-1}$ when
$\alpha$, $\beta$ and $x$ in Eq.~(\ref{eq:essentialfunction}) distributes uniformly. Although pixels of natural images follow Gaussian distribution, and
$J_1(k)$ and $J_2(k)$ in Eq.~(\ref{eq:DifferentialEquation}) do not distribute uniformly, we still can believe that the success probability of this method
is very high since only one bit satisfying the conditions in Eq.~(\ref{eq:determinebits1}) is needed, especially when $n$ is relatively large.

To verify the real performance of the above analysis, a number of experiments are carried out on some plain-images of size $512\times 512$ with the method \textit{W1)} when $n=32$. When $x_0=319684607/2^{32}$, $key_1=3835288501$, and $key_2=1437224678$. Two known plain-images ``Peppers" and ``Baboon", and the corresponding cipher-images are adopted.
Equivalent key $key1^*$, $key2^*$ and $\{B^*(k)\}_{k=0}^{8MN/n-1}$ is used to decrypt another cipher-image
shown in Fig.~\ref{fig:DecryptedLenna}a) and the recovered result is shown in Fig.~\ref{fig:DecryptedLenna}b), which is identical with the original version. From the
experiment, we found that only a little pixels (no more than ten pixels for plain-image of size $512\times 512$) are not recovered correctly when $n\leq 28$. This agree with our expectation as the probability that none of condition of Eq.~(\ref{eq:determineBk}) and condition~(\ref{eq:SpecialCondition}) are satisfied become larger when $n$ is smaller.

\begin{figure}[!htb]
\centering
\begin{minipage}{\figwidth}
\centering
\includegraphics[width=\textwidth]{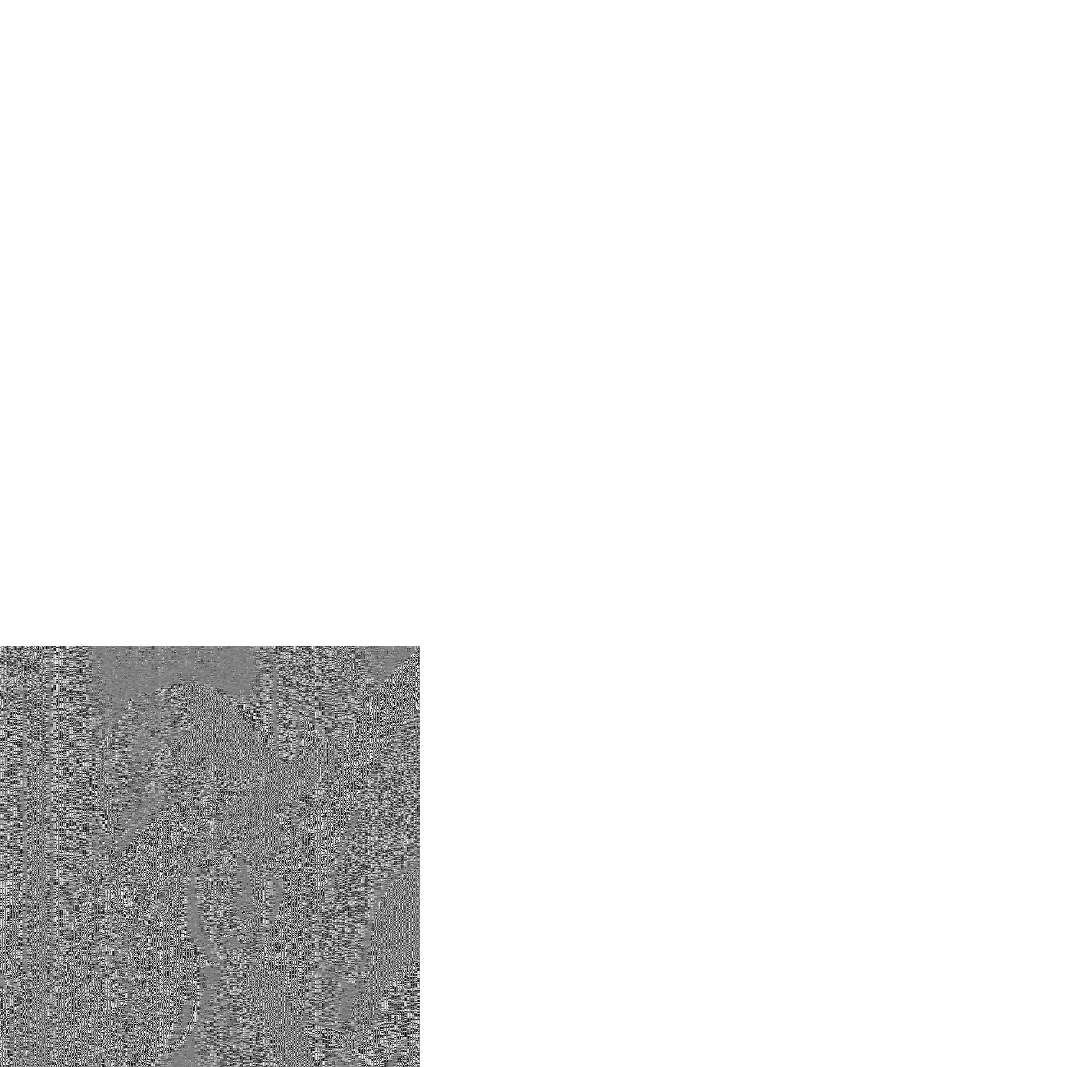}
a)
\end{minipage}
\begin{minipage}{\figwidth}
\centering
\includegraphics[width=\textwidth]{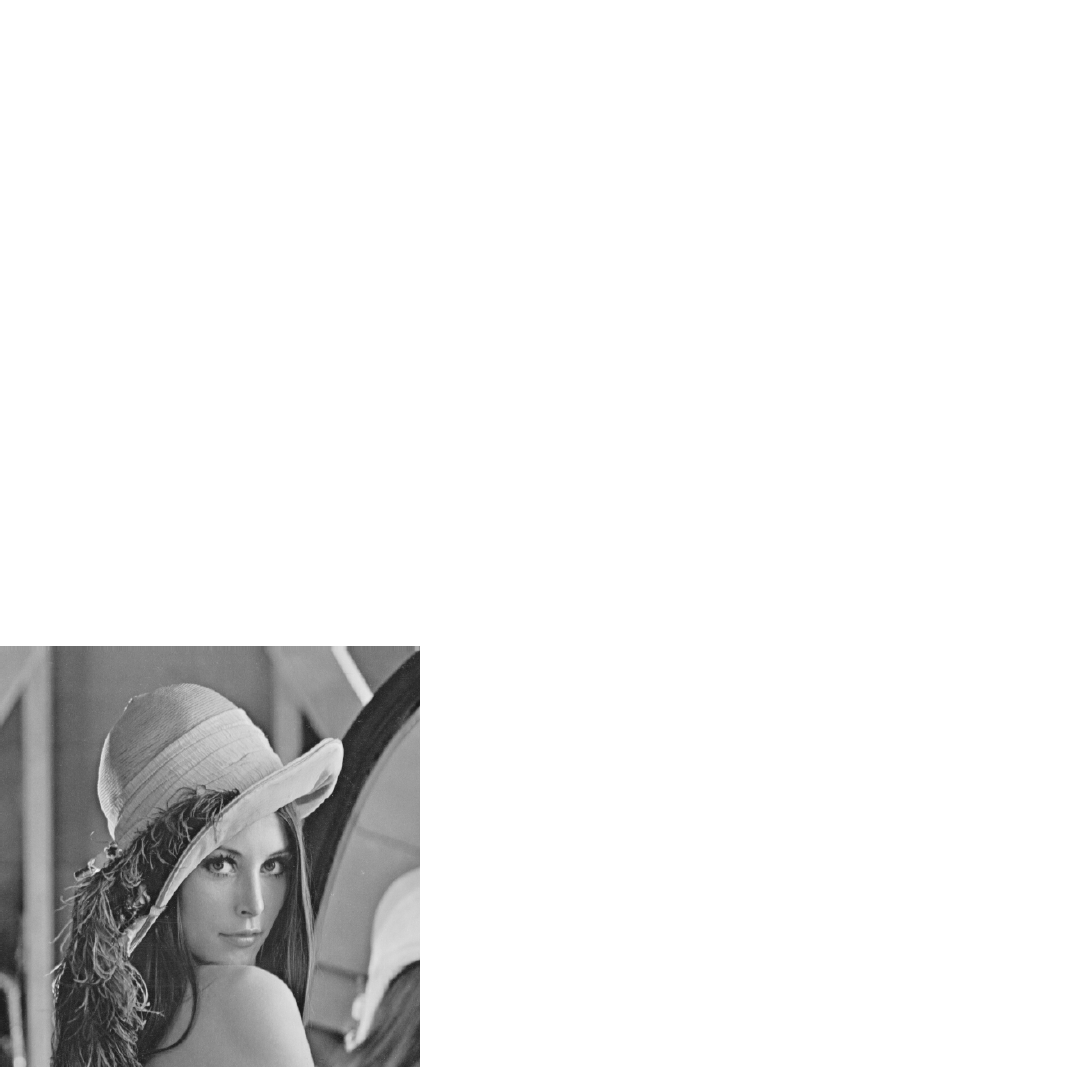}
b)
\end{minipage}
\caption{The decryption result of another cipher-image encrypted with the same secret key: a) cipher-image; b)
decrypted plain-image.}
\label{fig:DecryptedLenna}
\end{figure}

\subsection{Chosen-plaintext attack}

Chosen-plaintext attack is an enhanced version of known-plaintext attack, where the plaintext can be chosen arbitrarily to
obtain the information about the secret key in a more efficient way. In this subsection, the chosen-plaintext attack on MCKBA/HCKBA is briefly
introduced due to the following two points: 1) the known-plaintext attack on MCKBA/HCKBA works well in a relatively high probability and the chosen-plaintext
version can improve its performance a little; 2) the underlying theorem supporting the attack proposed in \cite[Theorem 1]{LCQ:MCKBA:IJBC11} is not right and corrected in Proposition~3.

\begin{proposition}
Assume that $\alpha, \beta, x$ are all $n$-bit integers, then a lower bound on the number of queries $(\alpha, \beta)$ to solve
Eq.~(\ref{eq:essentialfunction}) in terms of modulo $2^{n-1}$ for any $x$ is 1 if $n=2$; 2 if $n>2$.
\end{proposition}
\begin{proof}
When $n=2$, one can obtain $x_0=\tilde{y}_1$ by choosing $(\alpha_0, \beta_0)$ $=(1, 0)$.
When $n>2$, $\tilde{y}_1$ may be equal to zero or one no matter what $(\alpha_0, \beta_0)$ is, which means that
it is impossible to satisfy the condition of Property~3 for any $x$. So, we have to resort to another
query $(\alpha', \beta')$. Let $\alpha'_i, \beta'_i, y'_i, \tilde{y}'_{i}$ and $c'_i$ denote the counterparts of $\alpha_i, \beta_i, y_i, \tilde{y}_{i}$, and $c_i$
corresponding to $(\alpha', \beta')$. Given a set of $(\alpha_{i+k}, \beta_{i+k})$ and $(\alpha'_{i+k}, \beta'_{i+k})$, one can obtain
$(c_{i+k+1}, \tilde{y}_{i+k+1})$ and $(c'_{i+k+1}, \tilde{y}'_{i+k+1})$ from $(c_{i+k}, \tilde{y}_{i+k})$ and $(c'_{i+k}, \tilde{y}'_{i+k})$, respectively,
where $i, k$ are non-negative integers. Let arrows of plain head and ``V-back" head denote $x_{i+k}=0$ and $x_{i+k}=1$, respectively,
Fig.~\ref{fig:relationship} illustrates the mapping relationship between $(c_{i+k}, \tilde{y}_{i+k}, c'_{i+k}, \tilde{y}'_{i+k})$ and $(c_{i+k+1}, \tilde{y}_{i+k+1}, c'_{i+k+1}, \tilde{y}'_{i+k+1})$
for a given $(\alpha_{i+k}, \beta_{i+k}, \alpha'_{i+k}, \beta'_{i+k})$, where $k=0, 1$. Since $(c_{0}, \tilde{y}_{0}, c'_{0}, \tilde{y}'_{0})\equiv(0, 0, 0, 0)$,
the dashed arrows in Fig.~\ref{fig:relationship} describe operations of Eq.~(\ref{eq:essentialfunctionform}) in the two least significant bit planes corresponding to two sets of $(\alpha, \beta)$. Note that the data in the third column is exactly the same as the first one. Therefore, Fig.~\ref{fig:relationship} demonstrates operations of Eq.~(\ref{eq:essentialfunctionform})
under all different bit levels if the variable $i$ goes through $2\cdot t$, where $t=0\sim \lfloor n/2 \rfloor$ and $i+k\leq n-1$. Referring to Fig.~\ref{fig:relationship}, it can be
easily verified that $$1\in \{y_i, y'_i\}$$ is always satisfied, which means that $x_{i}$ can be derived from Table~1.
\end{proof}
\begin{figure}[!htb]
\centering
\begin{minipage}{1.5\figwidth}
\centering
\includegraphics[width=\textwidth]{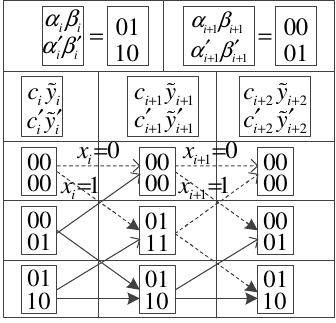}
a)
\end{minipage}
\caption{Relationship between $(c_{i+k}, \tilde{y}_{i+k},$ $c'_{i+k}, \tilde{y}'_{i+k})$ and $(c_{i+k+1}, \tilde{y}_{i+k+1},$ $c'_{i+k+1}, \tilde{y}'_{i+k+1})$ for
a given $(\alpha_{i+k},\beta_{i+k},$ $\alpha'_{i+k}, \beta'_{i+k})$, where $k=0, 1$.}
\label{fig:relationship}
\end{figure}

Under scenario of chosen-plaintext attack, one may make the plaintext satisfy that at least one pair of elements in $\{(J_1(k), J_2(k))\,|\,B(k)\in\{0, 1\}\}$ whose $i$-th bit plane satisfy the condition of Property~3. The same case exists for $\{(J_1(k), J_2(k))\,|\,B(k)\in\{2, 3\}\}$. The expected chosen-plaintext can be
obtained in a high probability by assigning $(J_1(k), J_2(k))$ with one of the two sets of number given in Corollary~\ref{coro:msb} randomly. Compared with
the known-plaintext attack, the chosen-plaintext attack has the following two superior performances:
1) the set $\{key1, key2\}$ can be reconstructed with much less complexity and much higher degree of accuracy;
2) the bits of $key(k)$ can be confirmed with a little higher probability.

\begin{corollary}
The $(n-1)$ least significant bits of $x$ in Eq.~(\ref{eq:essentialfunction}) can be determined easily
by setting $(\alpha, \beta)$ with the following two sets of numbers
\label{coro:msb}
\begin{eqnarray*}
\left\{ \left(\sum\nolimits_{j=0}^{\lceil n/2\rceil-1}(00)_2\cdot 4^j\right)\bmod 2^n, \left(\sum\nolimits_{j=0}^{\lceil n/2\rceil-1}(10)_2\cdot 4^j\right)\bmod 2^n\right\},\\
\left\{ \left(\sum\nolimits_{j=0}^{\lceil n/2\rceil-1}(10)_2\cdot 4^j\right)\bmod 2^n, \left(\sum\nolimits_{j=0}^{\lceil n/2\rceil-1}(01)_2\cdot 4^j\right)\bmod 2^n\right\},
\end{eqnarray*}
and checking the corresponding $\tilde{y}=y\oplus \alpha\oplus \beta$.
\end{corollary}
\begin{proof}
The proof is straightforward and therefore omitted.
\end{proof}

\section{Conclusion}

In this paper, the security of the image encryption algorithm MCKBA/HCKBA has been restudied
in detail. Based on some properties of a composite function composed of the modulo addition
and the XOR operation, a known-plaintext attack and an improved chosen-plaintext attack were provided
to determine an equivalent secret key of MCKBA/HCKBA. The cryptanalysis provided in
this paper sheds some light on breaking other encryption schemes based on multiple combination
of the modulo addition and XOR operations.

\section*{Acknowledgement}

This research was supported by the National Natural Science Foundation of China (No.~61100216), Scientific Research Fund of Hunan Provincial Education Department (No.~11B124), and Research Fund of Xiangtan University (No.~2011XZX16).

\bibliographystyle{ws-ijbc}
\bibliography{EMCKBA}
\end{document}